\documentclass{article}

\usepackage{arxiv}

\usepackage[utf8]{inputenc} 
\usepackage[T1]{fontenc}    
\RequirePackage{amsthm,amsmath,amsfonts,amssymb}
\RequirePackage{natbib}

\RequirePackage[colorlinks,citecolor=blue,urlcolor=blue]{hyperref}
\RequirePackage{graphicx}
\RequirePackage{tikz}

\usepackage{amsmath}
\usepackage{graphicx,psfrag,epsf}
\usepackage{enumerate}
\usepackage{natbib}
\setcitestyle{aysep={}} 
\usepackage{url} 
\usepackage{cleveref}
\usepackage{amssymb}
\RequirePackage{csquotes}
\usepackage{bbm}
\usepackage{mathtools}
\usepackage{xcolor}
\usepackage{multirow}
\usepackage{cancel}

\RequirePackage{csquotes}
\usepackage{bbm}
\usepackage{mathtools}
\graphicspath{ {./images/} }
\usepackage{setspace}
\theoremstyle{plain}

\newtheorem{theorem}{Theorem}[section]

\newtheorem{proposition}[theorem]{Proposition}

\theoremstyle{definition}
\newtheorem{definition}[theorem]{Definition}

\newtheorem{remark}[theorem]{Remark}

\title{An exhaustive ADDIS principle for online FWER control}

\author{
 Lasse Fischer \\
  Competence Center for Clinical Trials Bremen \\ University of Bremen \\
  \texttt{fischer1@uni-bremen.de} \\
   \And
 Marta Bofill Roig \\
  Center for Medical Data Science \\ Medical University of Vienna \\
  \texttt{marta.bofillroig@meduniwien.ac.at} \\
  \And
 Werner Brannath \\
  Competence Center for Clinical Trials Bremen \\ University of Bremen \\
  \texttt{brannath@uni-bremen.de} \\
}

\begin{document}
\maketitle
\begin{abstract}
In this paper we consider online multiple testing with familywise error rate (FWER) control, where the probability of committing at least one type I error shall remain under control while testing a possibly infinite sequence of hypotheses over time. Currently, Adaptive-Discard (ADDIS) procedures seem to be the most promising online procedures with FWER control in terms of power. Now, our main contribution is a uniform improvement of the ADDIS principle and thus of all ADDIS procedures. This means, the methods we propose reject as least as much hypotheses as ADDIS procedures and in some cases even more, while maintaining FWER control. In addition, we show that there is no other FWER controlling procedure that enlarges the event of rejecting any hypothesis. Finally, we apply the new principle to derive uniform improvements of the ADDIS-Spending and ADDIS-Graph.
\end{abstract}


\section{Introduction}
The decision in hypothesis testing is usually made by comparing a $p$-value with a prespecified significance level, e.g. $\alpha=0.05$. However, if not only one, but several hypotheses are tested, this will lead to an inflation of the type I error and some overall error control is needed. In life sciences, for example, it is often essential to avoid any type I error and therefore the probability of making at least one false discovery, called the familywise error rate (FWER), should be controlled under a previously defined level $\alpha \in (0,1)$. Classical multiple testing theory assumes that a  finite number of hypotheses $H_1,\ldots,H_m$, is predefined at the beginning of the evaluation \citep{bretz2016multiple}. A simple multiple testing procedure for controlling the FWER in this classical setting is, e.g.,  the Bonferroni correction, in which each individual hypothesis is tested at the level $\alpha / m$. In many contemporary trials, however, the hypothesis set grows over time and statistical inference is to be made with the already known hypotheses without knowing the future ones. In this paper, we focus on online multiple testing \citep{FS, JM}. Thereby, the hypotheses arrive one at a time and it must be decided on the current hypothesis while having access only to the previous hypotheses and decisions. Since the number of future hypotheses is also unknown in advance, it is usually assumed to be infinite.

Online multiple testing problems arise, for example, when public databases are used. Here, several researchers access the database at different times testing various hypotheses. In addition, many public databases grow over time as new data is collected, leading to the testing of further hypotheses \citep{robertson2022online}.
But comparatively smaller studies, such as specific platform trials, can also be formulated as an online multiple testing problem \citep{Retal}. Another interesting application is the sequential modification of a machine learning algorithm \citep{Fetal, FES}. Here, it is started with an initial model and at each step a hypothesis test is performed to decide whether a modification of the current model improves the performance. If the test is rejected, meaning the performance of the modified model is significantly larger, the modified model can be used as the new benchmark for future updates. In order to ensure with high probability that the benchmark model constantly improves, online FWER control is required.

\subsection{Problem formulation}
Let $(\Omega, \mathcal{A})$ be a measurable space and $\mathcal{P}$ be a set of probability distribution on $(\Omega, \mathcal{A})$. Unless otherwise stated, $\mathbb{P}\in \mathcal{P}$ denotes the true data generating distribution in the remainder of this paper. We have a sequence of null hypotheses $(H_i)_{i\in \mathbb{N}}$ about $\mathbb{P}$ with corresponding $p$-values $(P_i)_{i\in \mathbb{N}}$. The null $p$-values are assumed to be valid, meaning $\mathbb{P}(P_i\leq x)\leq x$ for all $x\in[0,1]$ and $i\in I_0$, where $I_0\subseteq \mathbb{N}$ is the index set of true null hypotheses. Note that this includes: uniformly distributed null $p$-values, where $\mathbb{P}(P_i\leq x)= x$ for all $x\in[0,1]$ and $i\in I_0$; and conservative null $p$-values, which means that $\mathbb{P}(P_i\leq x)< x$ for some $x\in [0,1]$. Let $V(i)$ denote the number of false rejections up to step $i\in \mathbb{N}$. The goal is to determine procedures that generate a sequence of individual significance levels $(\alpha_i)_{i\in \mathbb{N}}$, whereby each $\alpha_i$ is only allowed to depend on the previous $p$-values $P_1,\ldots,P_{i-1}$, such that 
\begin{align*}
    \text{FWER}(i)\coloneqq \mathbb{P}(V(i) > 0)  
\end{align*}
is controlled at level $\alpha\in (0,1)$ for each $i\in \mathbb{N}$.  $\mathbb{P}$ denotes the probability under the true configuration of true and false hypotheses. Since $\mathbb{P}(V(i) > 0)$ is increasing in $i$, it is sufficient to control $\text{FWER} \coloneqq \mathbb{P}(V > 0)$, where $V=\lim\limits_{i \to \infty} V(i)$. It is differentiated between \textit{strong} and \textit{weak} FWER control. Strong control provides that $\text{FWER}\leq \alpha$ under any configuration of true and false hypotheses, whereas weak control assumes that all null hypotheses are true. We focus on strong control. While controlling the FWER strongly, the \textit{power} should be maximised, where power is defined as the expected proportion of rejections among the false hypotheses.
In this paper, we aim for uniform improvements of existing procedures. By uniform improvement we mean that a procedure rejects as least as much hypotheses as the initial procedure and, in some cases, even more. This is the case, for example, when the newly constructed procedure tests the hypotheses at larger individual significance levels than the existing one for all data constellations. 

\subsection{Existing literature and contribution\label{sec:existing_lit}}
As a first FWER controlling online procedure, \cite{FS} introduced the Alpha-Spending as online version of the weighted Bonferroni, which sets the individual significance levels $(\alpha_i)_{i\in \mathbb{N}}$ such that $\sum_{i\in \mathbb{N}} \alpha_i \leq \alpha$. This method is conservative, meaning that there exist online FWER controlling procedures that uniformly improve Alpha-Spending. One approach to derive such improvements is the closure principle \citep{MPG}. One could either extend existing Bonferroni-based closed procedures to the online setting \citep{TR} or derive direct improvements of the Alpha-Spending via the online closure principle \citep{fischer2022online}. Most of these Alpha-Spending-based closed procedures can by summarised by choosing $(\alpha_i)_{i\in \mathbb{N}}$ such that $\sum_{i\in \mathbb{N}} \alpha_i (1-\mathbbm{1}_{P_i\leq \alpha_i}) \leq \alpha$. Here, if a hypothesis $H_i$ is rejected ($P_i\leq \alpha_i$), the significance level can be reused for future testing, which improves the classical Alpha-Spending. However, simulations have shown that these Alpha-Spending based closed procedures lead to low online power as well \citep{TR}. The problem is that the individual significance level $\alpha_i$ and thus the probability of reusing a significance level tend to zero for $i$ to $\infty$. A more promising approach is the ADaptive-DIScard (ADDIS) principle by \cite{TR}. The term ADDIS stems from the simultaneous discard of conservative null $p$-values based on a parameter $0<\tau_i\leq 1$ and adaption to the proportion of non-nulls ($p$-values corresponding to false null hypotheses) using a parameter $0\leq \lambda_i <\tau_i$.
 Precisely, an ADDIS procedure determines individual significance levels $(\alpha_i)_{i\in \mathbb{N}}$ such that 
\begin{align}
    \sum_{i\in \mathbb{N}} \frac{\alpha_i}{\tau_i - \lambda_i} (\mathbbm{1}_{P_i\leq \tau_i}-\mathbbm{1}_{P_i\leq \lambda_i}) \leq \alpha. \label{eq:addis_principle}
\end{align}
In contrast to Alpha-Spending-based closed procedures, ADDIS procedures allow to reuse the significance level $\alpha_i$ if either $P_i>\tau_i$ or $P_i\leq \lambda_i$. To correct for this improvement, the additional factor $1/(\tau_i-\lambda_i)$ needs to be included. Since null $p$-values are often conservative and thus tend to be large and non-null $p$-values tend to be small, we expect this trade-off to be useful. The interpretation of ADDIS procedures is as follows: Large $p$-values ($P_i>\tau_i$) are discarded, meaning not being tested, and small $p$-values ($P_i\leq \lambda_i$) are likely to be non-null and thus cannot lead to a type I error. However, to control the FWER, ADDIS procedures need additional assumptions. Firstly, the null $p$-values $(P_i)_{i\in I_0}$ need to be independent from each other and from the non-nulls. Secondly, in case of $\tau_i<1$, the null $p$-values are required to be uniformly valid, meaning $\mathbb{P}(P_i\leq xy|P_i\leq y)\leq x$ for all $x,y\in [0,1]$ and $i\in I_0$ \citep{ZSS}. In return,  ADDIS procedures lead to a high online power \citep{TR}. The ADDIS-Spending \citep{TR} and ADDIS-Graph \citep{fischer2023adaptive} were proposed as concrete ADDIS procedures satisfying the required conditions. 
 
 Although ADDIS procedures are quite powerful, they are based on the Bonferroni inequality \begin{align*}
\text{FWER}=P(V>0)=\mathbb{P}\left(\bigcup\limits_{i\in I_0} \{P_i \leq \alpha_i\} \right)  {\leq} \sum_{i\in I_0} \mathbb{P}({P_i \leq \alpha_i}), 
\end{align*}
 which leads to conservative procedures. It is known that under independence of $p$-values, the Bonferroni procedure can be uniformly improved by the Sidak correction \citep{vsidak1967rectangular}, which uses 
\begin{align} \text{FWER}=1-P(V=0)=1-\mathbb{P}\left(\bigcap\limits_{i\in I_0} \{P_i > \alpha_i\} \right)  = 1-\prod_{i\in I_0} \mathbb{P}({P_i \leq \alpha_i}). \label{eq:sidak}\end{align}
\cite{TR} have already attempted to apply the same idea to ADDIS procedures. However, since the individual significance level $\alpha_i$ of an ADDIS procedure uses information about the previous $p$-values, the events $\{P_i\leq \alpha_i\}$ are no longer independent of each other and thus the third equality in equation \eqref{eq:sidak} becomes an inequality \citep{TR}. This implies that the ADDIS-Sidak is conservative as well. Therefore, \cite{TR} left open the question of whether their ADDIS principle can be uniformly improved. In this paper, we will answer this question by introducing the exhaustive ADDIS principle, which provides a uniform improvement over the ADDIS principle by utilizing the independence of the p-values. In addition, we show that the there is no FWER controlling procedure that enlarges the event of rejecting any hypothesis further. 

\subsection{Overview of the paper}

In Section \ref{sec:addis_algorithm}, we introduce a general ADDIS algorithm that contains all other ADDIS procedures. Afterwards, we derive the exhaustive ADDIS algorithm as a uniform improvement of the ADDIS algorithm and show that the  event of rejecting any hypothesis cannot be further enlarged (Section \ref{sec:ex:addis_algorithm}). In Section \ref{sec:improved_addis_procedures}, this exhaustive ADDIS algorithm is used to obtain uniform improvements of the ADDIS-Spending and ADDIS-Graph. In Section \ref{sec:sim} and \ref{sec:real_data}, we quantify the performance of the constructed procedures by applying them on simulated and real data, respectively.  For the complete proofs of the theoretical results we refer to the Appendix and the \texttt{R} code for the simulations can be found at the GitHub repository \url{https://github.com/fischer23/Exhaustive-ADDIS-procedures}.


\section{The ADDIS algorithm\label{sec:addis_algorithm}}
In this section, we introduce a general ADDIS algorithm which encompasses all online procedures satisfying the ADDIS principle. This facilitates the interpretation of the ADDIS principle and conveniently introduces a notation needed to construct the uniform improvement in the next section.

In the most general form of the ADDIS principle, the parameters $\tau_i$ and $\lambda_i$ are allowed to depend on the previous $p$-values as well. Mathematically, $\alpha_i$, $\tau_i$ and $\lambda_i$ are random variables with values in $[0,1)$, $(0,1]$ and $[0,\tau_i)$ respectively that are measurable with respect to $\mathcal{G}_{i-1}\coloneqq \sigma(P_1,\ldots, P_{i-1})$. Furthermore, note that condition \eqref{eq:addis_principle} is equivalent to $\sum_{j=1}^{i} \frac{\alpha_j}{\tau_j-\lambda_j} (S_j-C_j) \leq \alpha$ for all $i\in \mathbb{N}$, where $S_i=\mathbbm{1}_{P_i\leq \tau_i}$ and $C_i=\mathbbm{1}_{P_i\leq \lambda_i}$. Since $\alpha_i$, $\tau_i$ and $\lambda_i$ are measurable with regard to $\mathcal{G}_{i-1}$, they must be fixed before knowing the true values of $S_i$ and $C_i$. Therefore, we need to make pessimistic assumptions at step $i\in \mathbb{N}$, meaning $S_i=1$ and $C_i=0$. Hence, condition \eqref{eq:addis_principle} is equivalent to \begin{align*}\sum_{j=1}^{i-1} \frac{\alpha_j}{\tau_j-\lambda_j} (S_j-C_j)+  \frac{\alpha_i}{\tau_i-\lambda_i}\leq \alpha \quad \forall i\in \mathbb{N}. \end{align*} With this, we can formulate a general ADDIS procedure, called ADDIS algorithm, that contains all other ADDIS procedures. 
 
 \begin{definition}[ADDIS algorithm\label{def:addis_algorithm}]\hphantom{1}
 \begin{enumerate}
 \setcounter{enumi}{-1}
 \item Before the study starts, choose $\tau_1\in (0,1]$, $\lambda_1 \in [0, \tau_1)$ and $\alpha_1 \in [0, \tau_1)$ such that \\ ${\alpha - \frac{\alpha_1 }{\tau_1 - \lambda_1} \geq 0}$.
 \item At step 1 reject $H_1$ if $P_1 \leq \alpha_1$ and set $\alpha^{(2)}=\alpha$ if $P_1 \leq \lambda_1$ or $P_1 > \tau_1$ and $\alpha^{(2)} = \alpha - \frac{\alpha_1 }{\tau_1-\lambda_1}$ if $ \lambda_1 < P_1 \leq \tau_1 $. Furthermore, choose for step $2$ some $\tau_2\in (0,1]$, $\lambda_2 \in [0, \tau_2)$ and $\alpha_2 \in [0, \tau_2)$ such that 
 $\alpha^{(2)} - \frac{\alpha_2 }{\tau_2-\lambda_2}\geq 0$. 
 \item[...]
 \item[i.] At step $i$ reject $H_i$ if $P_i \leq \alpha_i$ and set $\alpha^{(i+1)}=\alpha^{(i)}$ if $P_i \leq \lambda_i$ or $P_i > \tau_i$  and $\alpha^{(i+1)} = \alpha^{(i)} - \frac{\alpha_i }{\tau_i-\lambda_i}$ if $\lambda_i < P_i \leq \tau_i$. Furthermore, choose for step $i+1$ some $\tau_{i+1} \in (0,1]$, $\lambda_{i+1} \in [0, \tau_{i+1})$ and $\alpha_{i+1} \in [0, \tau_{i+1})$ such that 
 $\alpha^{(i+1)} - \frac{\alpha_{i+1}}{\tau_{i+1}-\lambda_{i+1}}\geq 0$.
 \end{enumerate}
 \end{definition}


For ease of understanding, we illustrate the ADDIS algorithm in Figure \ref{abb:addis_algorithm}. The chart includes the end of step $i-1$, i.e. the setting of $\tau_i$, $\lambda_i$ and $\alpha_i$, and the entire step $i$. The parameter $\alpha^{(i)}$, $i\in \mathbb{N}$, can be interpreted as the level that can (but  does not have to) be spent for the hypotheses $\{H_j:j\geq i\}$. Using Alpha-Spending one would set $\alpha^{(i+1)}=\alpha^{(i)}-\alpha_i$ for all $i\in \mathbb{N}$. However, as seen in Figure \ref{abb:addis_algorithm}, in the ADDIS algorithm the entire significance level is shifted to the future hypotheses ($\alpha^{(i+1)}=\alpha^{(i)}$) if  $P_i\leq \lambda_i$ or $P_i> \tau_i$ and in turn $\alpha^{(i+1)}=\alpha^{(i)}-\frac{\alpha_i}{\tau_i-\lambda_i}$ in the opposite case. 

 \begin{figure}[htbp]
 	\begin{center}
 			\centering
 		\includegraphics[width=19.5cm,height=6.5cm,keepaspectratio]{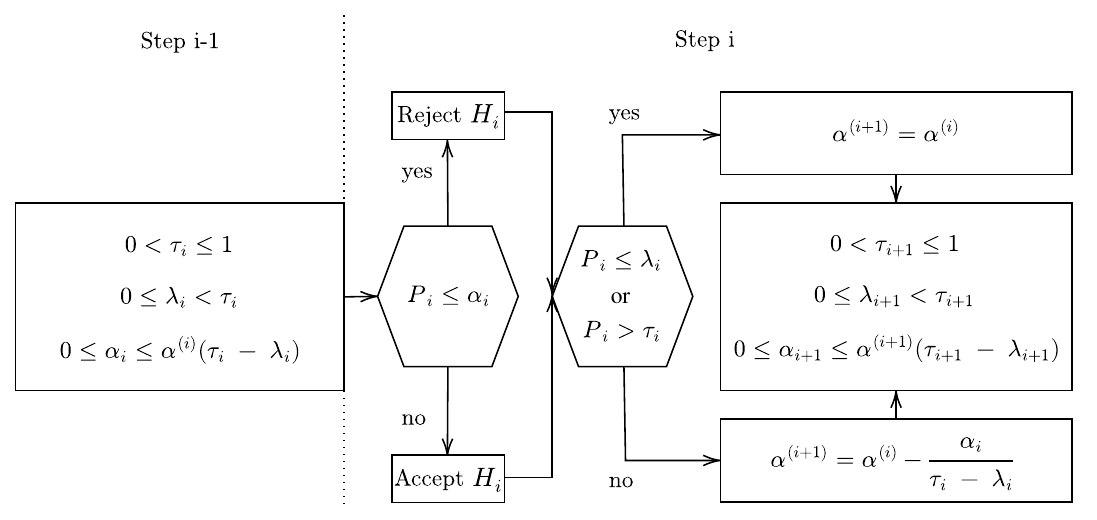}
 	\end{center}
 	\caption{Illustration of the ADDIS algorithm from the end of step $i-1$ to step $i$.\label{abb:addis_algorithm}}
 \end{figure}

\section{The exhaustive ADDIS algorithm\label{sec:ex:addis_algorithm}}

In this section, we introduce a uniform improvement of the ADDIS algorithm and show that it cannot be further improved. First, we illustrate the idea for the first two hypotheses and $\tau_1=1$. As in the ADDIS algorithm, we set $\alpha^{(2)}=\alpha^{(1)}=\alpha$ if $P_1\leq \lambda_1$. Now, under the global null hypothesis and the assumption of independent and uniformly distributed null $p$-values we can calculate
\begin{align*}
    \text{FWER}&=\mathbb{P}(P_1\leq \alpha_1 \cup P_2\leq \alpha_2) \\
    &= \mathbb{P}(P_1\leq \alpha_1 ) + \mathbb{P}(P_1 > \alpha_1, P_2\leq \alpha_2) \\
    &= \mathbb{P}(P_1\leq \alpha_1 ) +\mathbb{P}( P_2\leq \alpha_2| \alpha_1 < P_1 \leq \lambda_1) \mathbb{P}(  \alpha_1 < P_1 \leq \lambda_1) \\ &+ \mathbb{P}( P_2\leq \alpha_2|  P_1 > \lambda_1) \mathbb{P}( P_1 > \lambda_1) \\
    &\leq \mathbb{P}(P_1\leq \alpha_1 )+ \alpha \mathbb{P}(  \alpha_1 < P_1 \leq \lambda_1) + c \mathbb{P}( P_1 > \lambda_1),
\end{align*}
where $c$ denotes the value of $\alpha^{(2)}$ in case of $P_1>\lambda_1$. The last line follows from the independence of the null $p$-values and as we always choose $\alpha_2\leq \alpha^{(2)}$. Determining $c$ such that the last line equals $\alpha$ under the global null hypothesis and the assumption of uniformly distributed null $p$-values gives us $c=\alpha-\alpha_1\frac{1-\alpha}{1-\lambda}$.
This already indicates the uniform improvement over the ADDIS algorithm, as it sets $\alpha^{(2)}$ to the smaller value $\alpha^{(2)}=\alpha-\alpha_1\frac{1}{1-\lambda}\leq c$ in case of $P_1>\lambda_1$ (Figure \ref{abb:addis_algorithm}). In addition, in case of $\alpha_2=\alpha^{(2)}$ the last line in the above calculation becomes an equation and thus $\text{FWER}=\alpha$, which suggests that the procedure fully exhausts the significance level. In general, the \textit{exhaustive ADDIS algorithm} can be formulated as follows.

 \begin{definition}[Exhaustive ADDIS algorithm\label{def:ex_addis_algorithm}]\hphantom{1}
 \begin{enumerate}
 \setcounter{enumi}{-1}
 \item Before the study starts, choose $\tau_1\in (0,1]$, $\lambda_1 \in [\alpha \tau_1, \tau_1)$ and $\alpha_1 \in [0, \tau_1)$ such that \\ ${\alpha - \frac{\alpha_1(1-\alpha) }{\tau_1 - \lambda_1} \geq 0}$.
 \item At step 1 reject $H_1$ if $P_1 \leq \alpha_1$ and set $\alpha^{(2)}=\alpha$ if $P_1 \leq \lambda_1$ or $P_1 > \tau_1$ and $\alpha^{(2)} = \alpha - \frac{\alpha_1 (1-\alpha)}{\tau_1-\lambda_1}$ if $ \lambda_1 < P_1 \leq \tau_1 $. Furthermore, choose for step $2$ some $\tau_2\in (0,1]$, $\lambda_2 \in [\alpha^{(2)} \tau_2, \tau_2)$ and $\alpha_2 \in [0, \tau_2)$ such that 
 $\alpha^{(2)} - \frac{\alpha_2 (1-\alpha^{(2)})}{\tau_2-\lambda_2}\geq 0$. 
 \item[...]
 \item[i.] At step $i$ reject $H_i$ if $P_i \leq \alpha_i$ and set $\alpha^{(i+1)}=\alpha^{(i)}$ if $P_i \leq \lambda_i$ or $P_i > \tau_i$  and $\alpha^{(i+1)} = \alpha^{(i)} - \frac{\alpha_i (1-\alpha^{(i)})}{\tau_i-\lambda_i}$ if $\lambda_i < P_i \leq \tau_i$. Furthermore, choose for step $i+1$ some $\tau_{i+1} \in (0,1]$, $\lambda_{i+1} \in [\alpha^{(i+1)} \tau_{i+1}, \tau_{i+1})$ and $\alpha_{i+1} \in [0, \tau_{i+1})$ such that 
 $\alpha^{(i+1)} - \frac{\alpha_{i+1}(1-\alpha^{(i+1)})}{\tau_{i+1}-\lambda_{i+1}}\geq 0$.
 \end{enumerate}
 \end{definition}

  For proving that the exhaustive ADDIS algorithm controls the FWER, we use a backward induction similar to the calculation at the beginning of this section.
 \begin{theorem}\label{theo:strong_control}
 The exhaustive ADDIS algorithm controls the $\text{FWER}$ in the strong sense when the null $p$-values are uniformly valid, independent from each other and independent from the non-null p-values.\end{theorem}

 \begin{figure}[htbp]
 	\begin{center}
			\centering
 		\includegraphics[width=19.5cm,height=6.5cm,keepaspectratio]{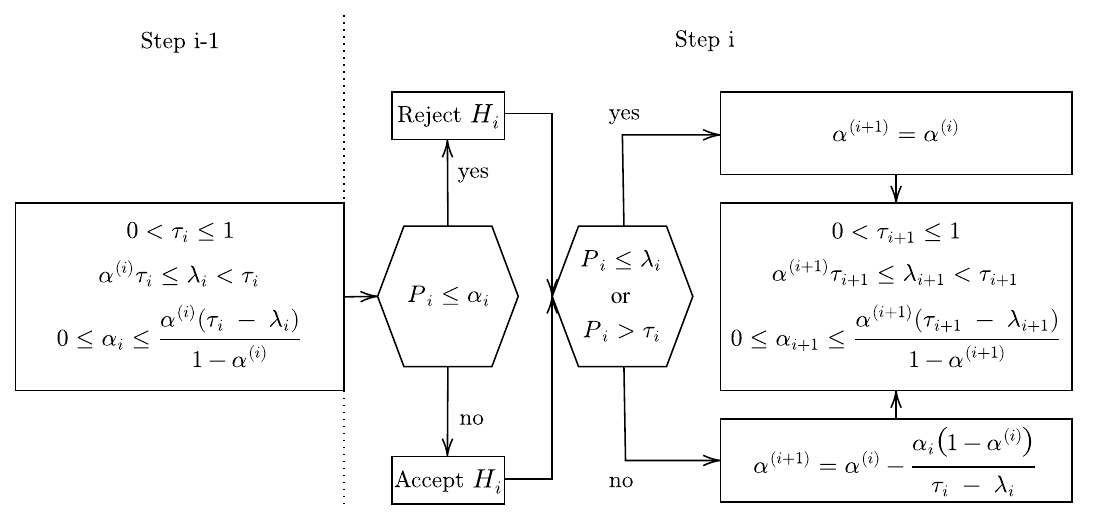}
 	\end{center}
 	\caption{Illustration of the exhaustive ADDIS algorithm from the end of step $i-1$ to step $i$.\label{abb:ex_addis_algorithm}}
 \end{figure}

 In Figure \ref{abb:ex_addis_algorithm}, the exhaustive ADDIS algorithm is illustrated. Note that the main difference between the exhaustive ADDIS algorithm and the ADDIS algorithm (Figure \ref{abb:addis_algorithm}) is the factor $(1-\alpha^{(i)})$ in $\alpha^{(i+1)} = \alpha^{(i)} - \frac{\alpha_i (1-\alpha^{(i)})}{\tau_i-\lambda_i}$ if $\lambda_i<P_i\leq \tau_i$. As the algorithm ensures that $\alpha^{(i)}\geq 0$ for all $i \in \mathbb{N}$, the level that can be spend for the future hypotheses is larger and therefore the exhaustive ADDIS algorithm is uniformly more powerful. Moreover, we have an additional constraint as we need to choose $\lambda_i \geq \tau_i \alpha^{(i)}$ for all $i\in \mathbb{N}$. However, one usually sets $\lambda_i \geq \tau_i \alpha^{(i)}$ anyway in order to exploit the potential of ADDIS procedures. To see this, note that closed procedures allow to reuse the significance level if $P_i\leq \alpha_i$ without requiring the factor $1/(\tau_i-\lambda_i)$ (see Section \ref{sec:existing_lit}). Therefore, one should always use $\lambda_i \gg \alpha_i$ in ADDIS procedures, which often leads to  $\lambda_i \geq \tau_i \alpha^{(i)}$ automatically. For example, \cite{TR} recommended to choose $\lambda_i=\tau_i \alpha$, which is always greater or equal than $\tau_i \alpha^{(i)} $. Hence, $\lambda_i \geq \tau_i \alpha^{(i)}$ is only a minor constraint.

\begin{remark}\hphantom{1}
An interesting special case of the exhaustive ADDIS algorithm is obtained for $\tau_i=1$ and $\lambda_i=\alpha^{(i)}$, where it basically reduces to the condition $\alpha_i\leq \alpha^{(i)}$ with $\alpha^{(i)}=\alpha-\sum_{j<i} \alpha_j (1-\mathbbm{1}_{P_j\leq \alpha^{(j)}})$. It is easy to see that this uniformly improves the Alpha-Spending-based closed procedures (see Section \ref{sec:existing_lit}). However, note that this improvement only works when the null $p$-values are independent of each other and the non-nulls.
\end{remark}

The exhaustive ADDIS algorithm does not only uniformly improve the ADDIS algorithm, but it is also optimal in the sense that the event of rejecting at least one hypothesis cannot be enlarged.  Before showing this, we prove that under the global null hypothesis the probability of committing any type I error is exactly $\alpha$.
\begin{proposition}\label{prop}
    Assume the null $p$-values are uniformly distributed and independent. In addition, let the individual significance levels of the exhaustive ADDIS algorithm be chosen such that $\sum_{i\in \mathbb{N}} \alpha_i \frac{1-\alpha^{(i)}}{\tau_i-\lambda_i} (S_i-C_i)=\alpha$. Then, under the global null hypothesis, we obtain:
    $$\mathbb{P}(V>0|I_0=\mathbb{N})=\alpha.$$
\end{proposition}
Note that under the global null hypothesis and uniformly distributed null p-values, there will almost surely be infinitely many $p$-values $P_j$ with $\lambda_j<P_j\leq \tau_j$, implying that $S_j-C_j=1$. Hence, the condition $\sum_{i\in \mathbb{N}} \alpha_i \frac{1-\alpha^{(i)}}{\tau_i-\lambda_i} (S_i-C_i)=\alpha$ can be ensured by construction of the procedure and does not depend on the data. In order to derive our optimality result, a weak assumption concerning the considered statistical model is required. Otherwise, \cite{GHS} showed that one can construct degenerate models in which no optimal tests exist. They introduced the following assumption to circumvent this issue:
\begin{align}
    \text{For any event } E\in \mathcal{A} \text{ with } \mathbb{P}(E|I_0=\mathbb{N})=0 \text{ it also holds } \mathbb{P}(E)=0. \label{eq:assump} 
\end{align}
This ensures that if the probability of an event under the global null hypothesis is zero, it is also zero under the true data constellation. This holds in most models considered in applied statistics \citep{GHS}. 
\begin{theorem}\label{theo:uniform_improvement}
    Assume that the assumptions of Proposition \ref{prop} are satisfied and \eqref{eq:assump} holds. Then there is no procedure $(\tilde{\alpha}_i)_{i\in \mathbb{N}}$ with FWER control such that 
    $$ \{\exists i\in \mathbb{N}: P_i\leq \tilde{\alpha}_i\} \supseteq \{\exists i\in \mathbb{N}: P_i\leq \alpha_i\} \text{ and }  \mathbb{P}(\{\exists i\in \mathbb{N}: P_i\leq \tilde{\alpha}_i\} \setminus \{\exists i\in \mathbb{N}: P_i\leq \alpha_i\})>0. $$ 
\end{theorem}
In particular, note that Theorem \ref{theo:uniform_improvement} does not hold for the ADDIS algorithm, as the exhaustive ADDIS algorithm is always as least as good as the ADDIS algorithm and leads to a larger probability of rejecting any hypothesis in some cases.

\section{Uniform improvements of ADDIS-Spending and ADDIS-Graph\label{sec:improved_addis_procedures}}

When applying the exhaustive ADDIS algorithm, one could just calculate the $\alpha^{(i)}$ at each step and then choose $\tau_i$, $\lambda_i$ and $\alpha_i$ according to the conditions given in Figure \ref{abb:ex_addis_algorithm}. However, in some situations the analyst may not want to do these steps manually and would like the algorithm to give out a specific significance level based on the previous test results. For example, this might be the case when the time between hypothesis tests is very short or when performing a simulation study. \cite{TR} proposed the ADDIS-Spending and \cite{fischer2023adaptive} the ADDIS-Graph as concrete procedures satisfying the ADDIS principle. In this Section, we show how these procedures could be improved using the exhaustive ADDIS algorithm.

 For the construction of the ADDIS-Spending and ADDIS-Graph it is started with a non-negative sequence $(\gamma_i)_{i\in \mathbb{N}}$ that sums to at most one and which can be interpreted as the initial allocation of the significance level $\alpha$. In order to obtain $\alpha^{(i+1)}=\alpha^{(i)}$ in case of $P_i\leq \lambda_i$ or $P_i>\tau_i$, the ADDIS-Spending ignores hypothesis $H_i$ in the future testing process, while the ADDIS-Graph distributes its level according to non-negative weights $(g_{j,i})_{i=j+1}^{\infty}$ that sum to at most one for each $j\in \mathbb{N}$. To compensate for the fact that we lose the level $\frac{\alpha_i}{\tau_i-\lambda_i}$ in case of $\lambda_i<P_i\leq \tau_i$ when using the ADDIS algorithm (Figure \ref{abb:addis_algorithm}), the significance levels of both procedures are multiplied by the factor $(\tau_i-\lambda_i)$. This leads to the individual significance level 
 \begin{align*}\alpha_i = (\tau_i-\lambda_i)  \alpha\gamma_{t(i)}, \quad \text{where } t(i)=1+ \sum_{j=1}^{i-1} (S_j-C_j).\end{align*}
 for the ADDIS-Spending and 
 \begin{equation}\label{eq:addis_graph}
\alpha_i = (\tau_i-\lambda_i)\left(\alpha \gamma_i + \sum_{j=1}^{i-1} g_{j,i}(C_j-S_j+1)  \frac{\alpha_j}{\tau_j-\lambda_j}\right)  
\end{equation}
for the ADDIS-Graph. We derive the graphical representation of the ADDIS-Graph in the Appendix.  

The easiest way to uniformly improve these procedures based on the exhaustive ADDIS algorithm is to change the factor $(\tau_i-\lambda_i)$ in both procedures to $(\tau_i-\lambda_i)/(1-\alpha^{(i)})$.
This leads to the \textit{Exhaustive-ADDIS-Spending} (E-ADDIS-Spending)
\begin{align*}
\alpha_i =  \frac{\tau_i-\lambda_i}{1-\alpha^{(i)}}\alpha \gamma_{t(i)}, \quad \text{where } t(i)=1+ \sum_{j=1}^{i-1} (S_j-C_j). 
\end{align*}
and \textit{Exhaustive-ADDIS-Graph} (E-ADDIS-Graph)
\begin{align*}
\alpha_i = \frac{\tau_i-\lambda_i}{1-\alpha^{(i)}}\left(\alpha \gamma_i + \sum_{j=1}^{i-1} g_{j,i} (C_j-S_j+1) \alpha_j \frac{1-\alpha^{(j)}}{\tau_j-\lambda_j}\right), 
\end{align*}
  where $\alpha^{(i)}=\alpha-\sum_{j=1}^{i-1} \frac{\alpha_j(1-\alpha^{(j)})}{\tau_j-\lambda_j} (S_j-C_j)$. Note that one needs to set $\lambda_i\geq \tau_i \alpha^{(i)}$ in these procedures as required by the exhaustive ADDIS algorithm.

  Since $\alpha^{(i)}\geq 0$ for all $i\in \mathbb{N}$ and at least $\alpha^{(1)}=\alpha > 0$, these procedures provide uniform improvements of the usual ADDIS procedures. However, except for unrealistic extreme cases, $\alpha^{(i)}$ will tend to zero for $i$ to infinity. This implies that these improvements can be very marginal for hypotheses that are tested at a late stage. For this reason, we propose a further approach to exploit the exhaustive ADDIS algorithm.
 
  Let $i\in \mathbb{N}$ be arbitrary. Note that 
$$ \sum_{j=1}^{ i}  \frac{\alpha_j(1-\alpha^{(j)})}{\tau_j-\lambda_j} (S_j-C_j) \leq \alpha \iff \sum_{j=1}^{ i}  \frac{\alpha_j}{\tau_j-\lambda_j} (S_j-C_j) \leq \alpha + \sum_{j=1}^{i}  \frac{\alpha_j \alpha^{(j)}}{\tau_j-\lambda_j} (S_j-C_j),$$
 where $\alpha^{(i)}=\alpha-\sum_{j=1}^{i-1} \frac{\alpha_j (1-\alpha^{(j)})}{\tau_j-\lambda_j}$ for all $i\in \mathbb{N}$.
 Thus, applying the ADDIS algorithm (Definition \ref{def:addis_algorithm}) at each step $i\in \mathbb{N}$ at the level $\alpha + \sum_{j=1}^{i}  \frac{\alpha_j \alpha^{(j)}}{\tau_j-\lambda_j} (S_j-C_j)$  is equivalent to applying the exhaustive ADDIS algorithm at level $\alpha$.
 
 For example, this could be incorporated into the ADDIS-Graph \eqref{eq:addis_graph} by distributing the level $(\alpha_j \alpha^{(j)})/(\tau_j-\lambda_j)$ in case of $\lambda_j < P_j \leq \tau_j$ to the future hypotheses according to non-negative weights $(h_{j,i})_{i=j+1}^{\infty}$ that sum to at most one for each $j\in \mathbb{N}$. This approach leads to the following procedure, which we term \textit{Evenly-Improved-ADDIS-Graph} (EI-ADDIS-Graph).
 \begin{align*}
\alpha_i = (\tau_i-\lambda_i)\left(\alpha \gamma_i + \sum_{j=1}^{i-1} g_{j,i}(C_j-S_j+1) \alpha_j \frac{1}{\tau_i-\lambda_j} + \sum_{j=1}^{i-1} h_{j,i}(S_j-C_j) \alpha_j \frac{\alpha^{(j)}}{\tau_i-\lambda_j}\right). 
\end{align*}

Note that the EI-ADDIS-Graph can be interpreted just as the ADDIS-Graph. However, the EI-ADDIS-Graph distributes the significance level to the future hypotheses also if $\lambda_j < P_j \leq \tau_j$, but reduced by the factor $\alpha^{(j)}$. Obviously, this defines a uniform improvement of the ADDIS-Graph. Furthermore, the weights $(h_{j,i})_{j\in \mathbb{N},i>j}$ determine which hypotheses benefit from the improvement, such that the gained significance level can be spent more evenly than using the E-ADDIS-Graph. This may lead to a larger power improvement when compared with the ADDIS-Graph. Also note that the same improvement could be applied to the ADDIS-Spending, as the ADDIS-Spending can be written as a specific ADDIS-Graph \citep{fischer2023adaptive}.

\section{Simulations\label{sec:sim}}
In this section, we aim to quantify the gain in power using the proposed EI-ADDIS-Graph instead of the ADDIS-Graph. To this regard, we consider the gaussian testing setup as described in \cite{TR}. More precise, $n$ null hypotheses $(H_i)_{i\in \{1,\ldots,n\}}$ of the form $H_i: \mu_i\coloneqq \mathbb{E}(Z_i)\leq 0$ are tested sequentially using z-tests. During the generation process of a test statistic $Z_i$, $i\in \{1,\ldots,n\}$, it is assumed that $Z_i=X_i+\mu_A$, $\mu_A\in \{2,4\}$, with probability $\pi_A\in \{0.1,\ldots, 0.9\}$ and $Z_i=X_i+\mu_N$, $\mu_N\in \{-2,0\}$, with probability $1-\pi_A$, where $X_1,\ldots,X_n\stackrel{i.i.d.}{\sim} N(0,1)$. Note that $\pi_A$ is the probability of a null hypothesis being false and $\mu_A$ the strength of the alternative. Moreover, the null $p$-values are conservative if $\mu_N=-2$ and uniformly distributed if $\mu_N=0$. 

 Figure \ref{fig:plot_qseries} and \ref{fig:plot_logq} show the estimated power and FWER of the Alpha-Spending, ADDIS-Graph and EI-ADDIS-Graph for $n=1000$ hypotheses by averaging over $2000$ independent trials (Alpha-Spending is included as a reference). The power is represented by the solid lines and the FWER by the dashed lines. In the left plots, the $p$-values are uniformly distributed ($\mu_N=0$) and in the right plots they are conservative ($\mu_N=-2$). As done by \cite{TR}, we applied the ADDIS procedures at level $\alpha=0.2$ with parameters $\tau_i=0.8$ and $\lambda_i=0.16$ for all $i\in \mathbb{N}$. Furthermore, we set $\gamma_i=\frac{6}{\pi^2 i^2}$ in Figure \ref{fig:plot_qseries} and  $\gamma_i\propto \frac{1}{(i+1)log(i+1)^2}$ in Figure \ref{fig:plot_logq} for all $i\in \mathbb{N}$. In addition, we chose $g_{j,i}=\gamma_{i-j}$ and $h_{j,i}=g_{j,i}$ for all $j\in \mathbb{N}$ and $i>j$ in all cases. 
 
 \begin{figure}[htbp]
	\begin{center}
			\centering		\includegraphics[width=18cm,height=6cm,keepaspectratio]{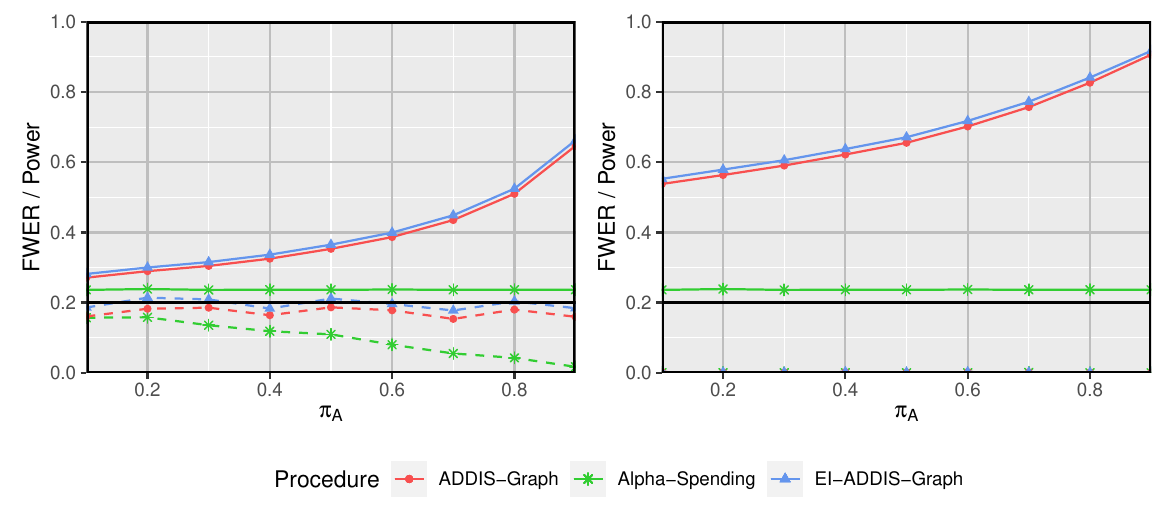}
	\end{center}
	\caption{Power and FWER for $n=1000$ hypotheses against proportion of false hypotheses ($\pi_A$) for Alpha-Spending, ADDIS-Graph and EI-ADDIS-Graph at level $\alpha=0.2$. Solid lines  correspond to power and dashed lines to FWER; Strength of the alternative is $\mu_A=4$ in both plots; $p$-values are uniformly distributed ($\mu_N=0$) in the left plot and conservative ($\mu_N=-2$) in the right; Procedures were applied with parameters $\gamma_i=\frac{6}{\pi^2 i^2}$, $g_{j,i}=\gamma_{i-j}$, $h_{j,i}=g_{j,i}$, $\tau_i=0.8$ and $\lambda_i=0.16$. \label{fig:plot_qseries}}
\end{figure}
 
 \begin{figure}[htbp]
	\begin{center}
			\centering
		\includegraphics[width=18cm,height=6cm,keepaspectratio]{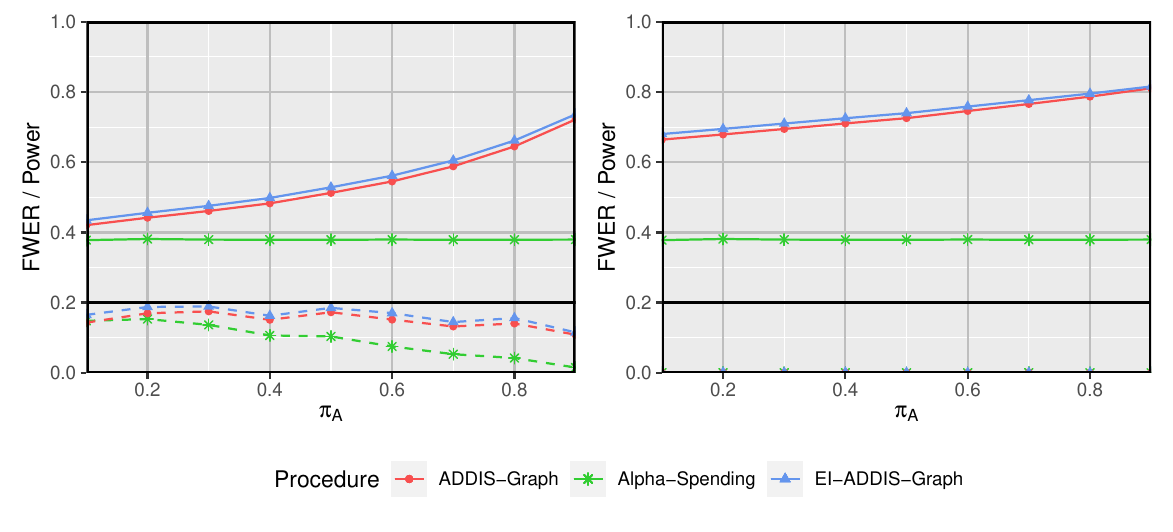}
	\end{center}
	\caption{Same caption as in Figure \ref{fig:plot_qseries}, except with $\gamma_i\propto \frac{1}{(i+1)log(i+1)^2}$}. \label{fig:plot_logq}
\end{figure}
 
The results show that the EI-ADDIS-Graph achieves a power improvement of $0.01$ to $0.02$ in all cases. Considering a number of $1000$ hypotheses this also means that we expect to make at least one more true discovery when using the EI-ADDIS-Graph instead of the ADDIS-Graph and even more when the proportion of false hypotheses is large. Furthermore, Figure \ref{fig:plot_qseries} indicates that the EI-ADDIS-Graph is roughly exhausting the significance level. Note that the level is not exhausted in Figure \ref{fig:plot_logq}, because the $(\gamma_i)_{i\in \mathbb{N}}$ decreases very slowly and since the testing process stops at step $1000$, not the entire level is used.

Since there are also online multiple testing problems, e.g. platform trials, where much less than one thousand hypotheses are tested, we also considered a setting with $n=10$ hypotheses. In this case, we also set the strength of the alternative to $\mu_A=2$. The results when applying the procedures with the same parameters as in Figure \ref{fig:plot_qseries} can be found in Figure \ref{fig:plot_logq_lown}. It is easy to see that in this case the power gain of the EI-ADDIS-Graph is larger. 

 \begin{figure}[htbp]
	\begin{center}
			\centering
		\includegraphics[width=18cm,height=6cm,keepaspectratio]{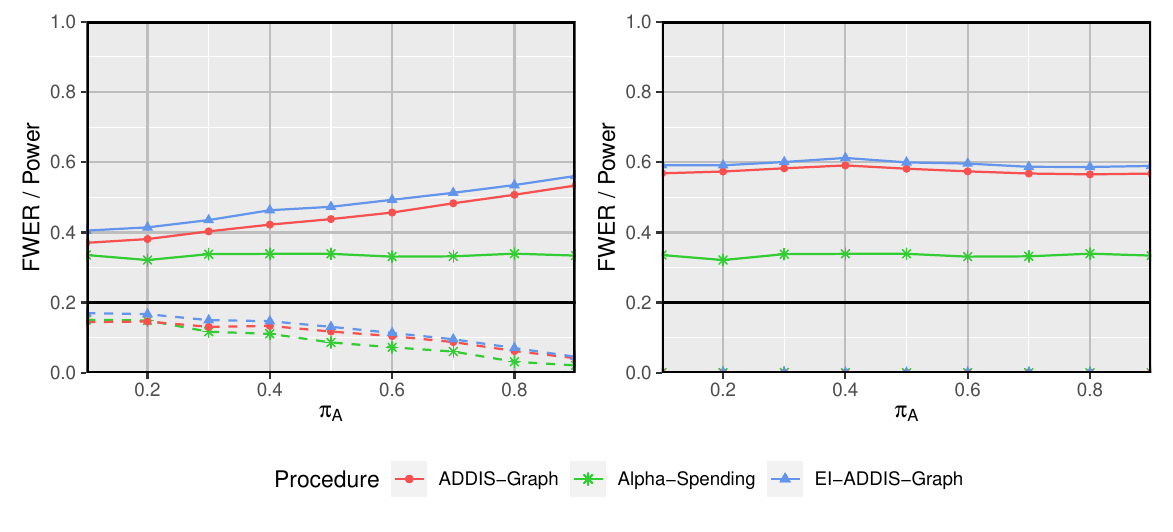}
	\end{center}
	\caption{Same caption as in Figure \ref{fig:plot_qseries}, except with $n=10$ and $\mu_A=2$.}\label{fig:plot_logq_lown}
\end{figure}

We do not present them here, but the same simulations were performed for E-ADDIS-Graph and E-ADDIS-Spending compared to ADDIS-Graph and ADDIS-Spending, respectively. As expected, the performance improvements are slightly smaller than those obtained using EI-ADDIS-Graph.

 \section{Application to IMPC data\label{sec:real_data}}
The International Mouse Phenotyping Consortium (IMPC) coordinates a large study to pin down the function of each protein-coding mouse gene \citep{mouse}. Since the resulting public data base is growing over time, IMPC data is used as a benchmark application for online multiple testing \citep{Retal}. We apply our procedures to 5000 of the $p$-values available at the Zenodo repository \url{https://zenodo.org/record/2396572} \citep{robertson2019onlinefdr}, which resulted from the evaluation by \cite{karp2017prevalence}. Note that the $p$-values in the IMPC data set are possibly correlated. However, we apply our procedures assuming independence for illustrative purposes.

 We applied Alpha-Spending, ADDIS-Graph, ADDIS-Spending, E-ADDIS-Graph, E-ADDIS-Spending and EI-ADDIS-Graph with the parameters $\gamma_i \propto \frac{1}{(i+1)log(i+1)^{1.5}}$, $g_{j,i}=\gamma_{i-j}$, $h_{j,i}=g_{j,i}$, $\tau_i=0.8$ and $\lambda_i=0.16$. The number of rejections obtained by the procedures for different FWER levels $\alpha$ can be found in Figure \ref{fig:plot_real}. As expected, the Alpha-Spending performed worst. In addition, the ADDIS-Graphs led to more rejections than the ADDIS-Spending procedures. The EI-ADDIS-Graph led to most rejections, which is consistent with our theoretical reasoning that the EI-ADDIS-Graph performs best in practice. In all cases, it rejected more hypotheses than the ADDIS-Graph. The rejection gap between these procedure is growing for an increasing FWER level. At level $\alpha=0.05$ the EI-ADDIS-Graph rejected $3$ hypotheses more than the ADDIS-Graph, while it led to $25$ additional rejections at $\alpha=0.4$.

 \begin{figure}[htbp]
 	\begin{center}
			\centering
 		\includegraphics[width=18cm,height=6cm,keepaspectratio]{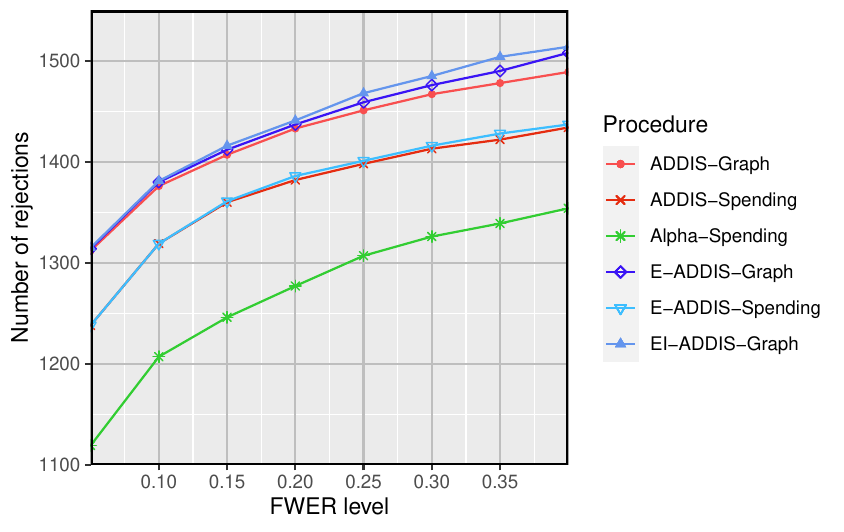}
 	\end{center}
 	\caption{Number of rejections obtained by applying Alpha-Spending, ADDIS-Graph, ADDIS-Spending, E-ADDIS-Graph, E-ADDIS-Spending and EI-ADDIS-Graph to IMPC data for different FWER levels. Procedures were applied with parameters $\gamma_i \propto \frac{1}{(n+1)log(n+1)^{1.5}}$, $g_{j,i}=\gamma_{i-j}$, $h_{j,i}=g_{j,i}$, $\tau_i=0.8$ and $\lambda_i=0.16$.}
 	\label{fig:plot_real}
 \end{figure}

\section{Discussion}

\cite{TR} asked whether their ADDIS principle is uniform improvable. We answered this question by firstly formulating the ADDIS principle as a general online procedure (Section \ref{sec:addis_algorithm}), called ADDIS algorithm, and secondly introducing an exhaustive ADDIS algorithm (Section \ref{sec:ex:addis_algorithm}), which is a uniform improvement of the ADDIS algorithm.  Since no additional assumptions are needed for obtaining the improvement, the exhaustive ADDIS procedures should be preferred over the usual ADDIS procedures in practice. Furthermore, the proposed methods not only lead to power improvements, but we also provide a general algorithm for their construction which is easy to use.

In practical use, offering software has become increasingly necessary. Statistical software, such as \texttt{R} packages and shiny apps, have been introduced to implement online control procedures \citep{robertson2019onlinefdr}. As part of future efforts, the creation of a new package or the inclusion of new functions in existing software that implement our proposed procedures will be considered.

Our uniform improvement is based on a different proof idea than the usual ADDIS principle,  which makes it possible to exploit the independence between $p$-values. We wonder whether the same approach can be used to derive adaptive procedures that work under more complex (but still known) dependence structures. This could be addressed in future work.

\newpage

\section*{Appendix}\label{appn}

\subsection*{Proofs}

 \begin{proof}[Proof of Theorem \ref{theo:strong_control}]
 We first show weak FWER control. Strong FWER is then obtained by the (online) closure principle. For proving weak control it is sufficient to show that $\text{FWER}\leq \alpha$ given $I_0=\mathbb{N}$. Let $R_i=\mathbbm{1}_{P_i \leq \alpha_i}$ and $R_{j:i}=\sum_{k=j}^i \mathbbm{1}_{P_k \leq \alpha_k}$ for $i,j\in \mathbb{N}$, $i\geq j$. Then weak control of $\text{FWER}$ is equivalent to $\mathbb{P}(R_{1:i} >0|I_0=\mathbb{N}) \leq \alpha$ for all $i\in \mathbb{N}$ (we omit $|I_0=\mathbb{N}$ in the remaining proof for better readability). 

 Let $i \in \mathbb{N}$ be arbitrary but fixed. We first use backwards induction to deduce that 
 $$\mathbb{P}(R_{j:i}>0 | \mathcal{G}_{j-1}) \leq \alpha^{(j)} \qquad \forall j \in \{1, \ldots,i\},$$
 where $\mathcal{G}_j=\sigma(\{P_1, \ldots, P_j\})$.\newline
 \begin{itemize}
     \item
 Initial Case ($j=i$):
 $$\mathbb{P}(R_{i:i} >0 | \mathcal{G}_{i-1})=\mathbb{P}(R_i >0| \mathcal{G}_{i-1}) = \mathbb{P}(P_i \leq \alpha_i| \mathcal{G}_{i-1})\leq \alpha_i \leq \frac{\alpha^{(i)}(\tau_i- \lambda_{i})}{1-\alpha^{(i)}} \leq \alpha^{(i)},$$
 since $\lambda_i \geq \tau_i \alpha^{(i)}$ and $\tau_i\leq 1$ for all $i \in \mathbb{N}$. \newline
 \item Induction Hypothesis (IH): We assume that for an arbitrary $j\in \{2,\ldots,i\}:$ $\mathbb{P}(R_{j:i}>0| \mathcal{G}_{j-1}) \leq \alpha^{(j)}.$ In the following, we also use that  $\alpha^{(j)}=\alpha^{(j-1)}$ if $P_{j-1} \leq \lambda_{j-1}$ or $P_{j-1} > \tau_{j-1}$ and $\alpha^{(j)}=\alpha^{(j-1)}- \frac{\alpha_{j-1} (1- \alpha^{(j-1)})}{\tau_{j-1}-\lambda_{j-1}}$ if $ \lambda_{j-1}< P_{j-1} \leq \tau_{j-1}$.\newline
 \item Induction Step ($ j \rightarrow j-1$):
 \begin{align*}
\mathbb{P}(R_{(j-1):i} >0| \mathcal{G}_{j-2})&= \mathbb{P}(R_{j-1}>0| \mathcal{G}_{j-2})+\mathbb{P}(R_{j:i}>0, R_{j-1}=0| \mathcal{G}_{j-2}) \\
 &=\mathbb{P}(R_{j-1} >0| \mathcal{G}_{j-2}) \\
 &+ \mathbb{P}(R_{j:i} >0|  \lambda_{j-1}< P_{j-1} \leq \tau_{j-1}, \mathcal{G}_{j-2}) \mathbb{P}(\lambda_{j-1}< P_{j-1} \leq \tau_{j-1}| \mathcal{G}_{j-2}) \\ &+ \mathbb{P}(R_{j:i} >0 | \alpha_{j-1} < P_{j-1} \leq \lambda_{j-1} ,\mathcal{G}_{j-2}) \mathbb{P}(\alpha_{j-1} < P_{j-1} \leq \lambda_{j-1} | \mathcal{G}_{j-2})\\ &+ \mathbb{P}(R_{j:i} >0 | P_{j-1} > \tau_{j-1} ,\mathcal{G}_{j-2}) \mathbb{P}(P_{j-1} > \tau_{j-1} | \mathcal{G}_{j-2})\\
 & \stackrel{\text{IH}}{\leq} \mathbb{P}(P_{j-1} \leq \alpha_{j-1}| \mathcal{G}_{j-2}) \\ &+ \left(\alpha^{(j-1)}- \frac{\alpha_{j-1} (1- \alpha^{(j-1)})}{\tau_{j-1}-\lambda_{j-1}}\right) \mathbb{P}(\lambda_{j-1}< P_{j-1} \leq \tau_{j-1} | \mathcal{G}_{j-2}) \\ &+  \alpha^{(j-1)}\mathbb{P}(\alpha_{j-1} < P_{j-1} \leq \lambda_{j-1}| \mathcal{G}_{j-2})+\alpha^{(j-1)}\mathbb{P}( P_{j-1} > \tau_{j-1}| \mathcal{G}_{j-2}) \\
 &= (1-\alpha^{(j-1)}) \mathbb{P}(P_{j-1} \leq \alpha_{j-1}| \mathcal{G}_{j-2}) \\& +   \frac{\alpha_{j-1} (1- \alpha^{(j-1)})}{\tau_{j-1}-\lambda_{j-1}} \mathbb{P}(P_{j-1} \leq \lambda_{j-1}| \mathcal{G}_{j-2}) \\ &+\alpha^{(j-1)}- \frac{\alpha_{j-1} (1- \alpha^{(j-1)})}{\tau_{j-1}-\lambda_{j-1}}\mathbb{P}(P_{j-1} \leq \tau_{j-1}| \mathcal{G}_{j-2}) \\
 &= \alpha^{(j-1)} \\ 
 &+ \mathbb{P}(P_{j-1} \leq \tau_{j-1}| \mathcal{G}_{j-2}) (1- \alpha^{(j-1)}) \left[\vphantom{\frac{\alpha_{j-1} }{\tau_{j-1}-\lambda_{j-1}}} \mathbb{P}(P_{j-1} \leq \alpha_{j-1}| P_{j-1}\leq \tau_{j-1}, \mathcal{G}_{j-2}) \right. \\ 
 &+ \left. \frac{\alpha_{j-1} }{\tau_{j-1}-\lambda_{j-1}} \mathbb{P}(P_{j-1} \leq \lambda_{j-1}| P_{j-1}\leq \tau_{j-1}, \mathcal{G}_{j-2})-\frac{\alpha_{j-1} }{\tau_{j-1}-\lambda_{j-1}}\right]
 \\
 & \leq \alpha^{(j-1)}  + \mathbb{P}(P_{j-1} \leq \tau_{j-1}| \mathcal{G}_{j-2}) (1-\alpha^{(j-1)}) \alpha_{j-1} \\ 
 & \left[\frac{1}{\tau_{j-1}} + \frac{\lambda_{j-1}}{\tau_{j-1}(\tau_{j-1}-\lambda_{j-1})} - \frac{1}{\tau_{j-1}-\lambda_{j-1}}\right] \\
 &=\alpha^{(j-1)}. \end{align*}
 Note that the independence of the null $p$-values was used in both inequalities and the uniformly validity  was applied in the second inequality. Since $\alpha^{(1)} = \alpha$ and $\mathcal{G}_{0}=\emptyset$, in particular $\mathbb{P}(R_{1:i} >0) \leq \alpha$. As $i \in \mathbb{N}$ was arbitrary, weak FWER control follows.

 \end{itemize}
Now we use this to show strong control. To this end, let $I_0 \subseteq \mathbb{N}$ be arbitrary and $(\alpha_i)_{i\in \mathbb{N}}$ be obtained by applying the exhaustive ADDIS algorithm to $(H_i)_{i\in \mathbb{N}}$. Since the null $p$-values are independent from the non-nulls and $$\sum_{j<i, j\in I_0} \alpha_j \frac{1-\alpha_{I_0}^{(j)}}{\tau_j-\lambda_j}\leq \sum_{j<i} \alpha_j \frac{1-\alpha^{(j)}}{\tau_j-\lambda_j}\leq \alpha \quad \text{ for every } i\in \mathbb{N}, $$  where $\alpha_{I_0}^{(j)}=\alpha - \sum_{k<j, j\in I_0} \alpha_k \frac{1-\alpha_{I_0}^{(k)}}{\tau_k-\lambda_k}$, the levels $(\alpha_i)_{i\in I_0}$ could also be obtained by applying the exhaustive ADDIS algorithm to $H_{I_0}$. Consequently, strong FWER control is implied by weak control.
 \end{proof}

\begin{proof}[Proof of  Proposition \ref{prop}]
 The assertion follows from the proof of Theorem \ref{theo:strong_control}. For this, first consider the case where the entire significance level is spent for a finite number of hypotheses, thus $\alpha^{(i)}=0$ for some $i\in \mathbb{N}$. Then $\alpha_i=\alpha^{(i)}$ and we obtain an equality in the initial case of the induction proof. Together with the assumption of uniformly distributed null $p$-values, all inequalities of the induction proof become equations and we obtain that the probability of rejecting any hypothesis is exactly $\alpha$. Now consider the case where $\alpha^{(i)}>0$ for all $i\in \mathbb{N}$ and let $\delta >0$ be arbitrary. Since we assume $\sum_{i\in \mathbb{N}} \alpha_i \frac{1-\alpha^{(i)}}{\tau_i-\lambda_i}(S_i-C_i)=\alpha$, and, by the exhaustive algorithm
$$
\alpha^{(i)}=\alpha-\sum_{j=1}^{i-1} \alpha_j \frac{1-\alpha^{(j)}}{\tau_j-\lambda_j}(S_j-C_j),
$$
we obtain
$\alpha^{(i)}\to 0$ for $i\to \infty$. Hence, there exists an $i\in \mathbb{N}$ such that $\alpha_i \geq 0 \geq  \alpha^{(i)}-\delta$. Thus,$\mathbb{P} (R_{i:i}>0|\mathcal{G}_{i-1}, I_0=\mathbb{N})= \alpha_i\geq \alpha^{(i)} - \delta$. With an analogous induction proof we obtain $\mathbb{P} (R_{j:i}>0|\mathcal{G}_{j-1}, I_0=\mathbb{N})\geq \alpha^{(j)} - \delta$ for all $j\leq i$ and hence  $\mathbb{P}(V(i)>0| I_0=\mathbb{N})\geq \alpha - \delta$. Since $\delta >0$ was arbitrary, $\mathbb{P}(V(i)>0| I_0=\mathbb{N})\to \alpha$ for $i \to \infty$.

\end{proof}

\begin{proof}[Proof of Theorem \ref{theo:uniform_improvement}]
    Together with Proposition \ref{prop} and \eqref{eq:assump}, the Theorem follows by Proposition 3.8 in \cite{fischer2022online}. However, we also provide a self-contained proof: Suppose there exists a procedure $(\tilde{\alpha}_i)_{i\in \mathbb{N}}$ with the  property stated in the Theorem. Then, we have
    \begin{align*}\mathbb{P}(\{\exists i\in \mathbb{N}: P_i\leq \tilde{\alpha}_i\}|I_0=\mathbb{N})&=\mathbb{P}(\{\exists i\in \mathbb{N}: P_i\leq \tilde{\alpha}_i\}\setminus \{\exists i\in \mathbb{N}: P_i\leq \alpha_i\}|I_0=\mathbb{N}) \\
    &+\mathbb{P}(\{\exists i\in \mathbb{N}: P_i\leq \tilde{\alpha}_i\} \cap \{\exists i\in \mathbb{N}: P_i\leq \alpha_i\}|I_0=\mathbb{N}) \\
    &> \mathbb{P}(\{\exists i\in \mathbb{N}: P_i\leq \alpha_i\}|I_0=\mathbb{N})= \alpha, 
    \end{align*}
    which contradicts the assumption that $(\tilde{\alpha}_i)_{i\in \mathbb{N}}$ controls the FWER. The inequality follows by \eqref{eq:assump} and the last equality by Proposition \ref{prop}.
\end{proof}

\subsection*{Graphical representation of the ADDIS-Graph}

\cite{BWBP} proposed the representation of Bonferroni-based closed procedures as directed graphs, where the hypotheses are nodes, connected by weighted vertices illustrating the distribution of the significance level in case of a rejection. Such graphical procedures are becoming increasingly popular, since they make the calculation of individual significance levels easy to follow and thus facilitate communication with users. \cite{fischer2023adaptive} showed that the ADDIS-Graph can be represented in a similar way. 

For this, it makes sense to first define $$\tilde{\alpha}_i=\frac{\alpha_i}{\tau_i-\lambda_i}=\alpha \gamma_i + \sum_{j=1}^{i-1} g_{j,i} (C_j-S_j+1) \tilde{\alpha}_j,$$
where $\alpha_i$ denotes the individual significance level of the ADDIS-Graph.
Note that the levels $\tilde{\alpha}_i$, $i\in \mathbb{N}$, can also be obtained by the following recursion formula:
\begin{enumerate}
    \item Set $\tilde{\alpha}_i=\alpha \gamma_i$ for all $i\in \mathbb{N}$ and $j=1$.
    \item If $P_j\leq \lambda_j$ or $P_j> \tau_j$, update the levels $\tilde{\alpha}_i$, $i>j$, as follows: $\tilde{\alpha}_i\to\tilde{\alpha}_i+g_{j,i}\tilde{\alpha}_j$. \label{2}
    \item Update $j\to j+1$ and go to step \ref{2}.
\end{enumerate}
In Figure \ref{fig:graph}, we used this to represent the levels $(\tilde{\alpha}_i)$ as a graph. The initial levels are illustrated below the hypotheses. In case of $P_j>\tau_j$ or $P_j\leq \lambda_j$, the level $\tilde{\alpha}_j$ is distributed to the future hypotheses according to the weights $(g_{j,i})_{j\in \mathbb{N},i>j}$. After applying this graph, the significance level of the ADDIS-Graph can be obtained by taking $\alpha_i=\tilde{\alpha}_i(\tau_i-\lambda_i)$.

\begin{figure}[htbp]
	\begin{center}
			\centering
   \includegraphics[width=16.5cm,height=5.5cm,keepaspectratio]{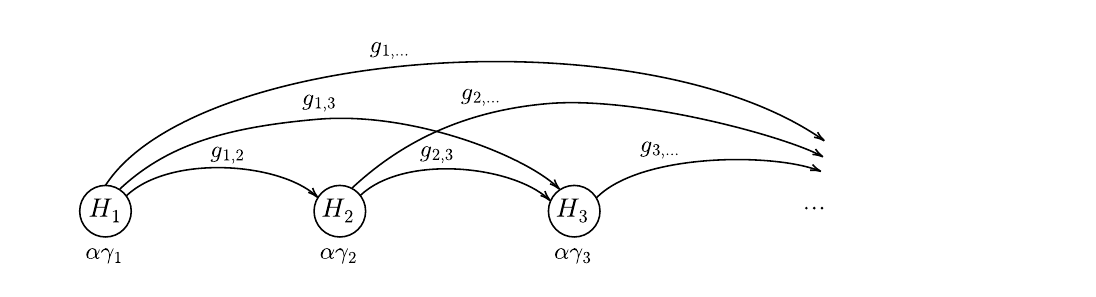}
	\end{center}
	\caption{Graphical representation of the levels $(\tilde{\alpha}_i)_{i\in \mathbb{N}}$, where $\tilde{\alpha}_i=\alpha_i/(\tau_i-\lambda_i)$ and $\alpha_i$ is the level of the ADDIS-Graph. The initial levels are illustrated below each node. In case of $P_i>\tau_i$ or $P_i\leq \lambda_i$, the future significance levels are updated according to the weights $(g_{j,i})_{j\in \mathbb{N}, i>j}$. \label{fig:graph}}
\end{figure}

\section*{Acknowledgements}
 The authors are grateful for the valuable comments of two anonymous referees.

\section*{Funding} 
L. Fischer acknowledges funding by the Deutsche Forschungsgemeinschaft (DFG, German Research Foundation) – Project number 281474342/GRK2224/2.

M. Bofill Roig is a member of the EU Patient-centric clinical trial platform (EU-PEARL). EU-PEARL has received funding from the Innovative Medicines Initiative 2 Joint Undertaking under grant agreement No 853966. This Joint Undertaking receives support from the European Union's Horizon 2020 research and innovation programme and EFPIA and Children's Tumor Foundation, Global Alliance for TB Drug Development non-profit organization, Spring- works Therapeutics Inc. This publication reflects the author's views. Neither IMI nor the European Union, EFPIA, or any Associated Partners are responsible for any use that may be made of the information contained herein.

\bibliographystyle{apalike}  
\bibliography{Bibliography-MM-MC.bib}

\end{document}